\documentclass[a4paper,UKenglish,cleveref, autoref, thm-restate]{lipics-v2021}

\hideLIPIcs

\usepackage{bussproofs}
\usepackage{xspace}
\usepackage{bm}

\newcommand{\Red}{\Math{\boldsymbol{\forall}\kern .04ex\ProofSystemsFont{red}}}

\newcommand{\red}{\Math{\,\raisebox{.2ex}{$\scriptstyle+$}\,\Red}}

\newlength{\breite}
\newcommand{\Math}[1]{\ensuremath{#1}}
\newcommand{\ComplexityClassFont}[1]{\mathsf{#1}}

\newcommand{\ProofSystemsFont}[1]{\mathsf{#1}}
\newcommand{\DefineProofSystem}[2]{\expandafter\def\csname#1\endcsname{\Math{\ComplexityClassFont{#2}}\xspace}}
\newcommand{\DefineProofSystemNoSpace}[2]
{\expandafter\def\csname#1\endcsname{\Math{\ComplexityClassFont{#2}}}}
\newcommand{\Frege}[1][]{\Math{\ifthenelse{\isempty{#1}}
{\ProofSystemsFont{Frege}}{#1\text{-}\ProofSystemsFont{Frege}}}}
\DefineProofSystem{eFrege}{e\Frege[]\xspace}
\newcommand{\eFregeRed}[1][]{\eFrege\!\!\red}
\DefineProofSystemNoSpace{DRAT}{\ProofSystemsFont{DRAT}}
\DefineProofSystem{MCPS}{\ProofSystemsFont{CLIP}}
\DefineProofSystemNoSpace{clip}{\ProofSystemsFont{CLIP}}
\DefineProofSystem{micep}{\ProofSystemsFont{MICE'}}
\DefineProofSystem{cpog}{\ProofSystemsFont{CPOG}}
\DefineProofSystem{kcps}{\ProofSystemsFont{KCPS(\#SAT)}}
\DefineProofSystem{mice}{\ProofSystemsFont{MICE}}
\DefineProofSystem{qmice}{\ProofSystemsFont{Q\text{-}MICE}}
\DefineProofSystem{emice}{\Math{\forall}\ProofSystemsFont{Exp\text{+}MICE}}
\DefineProofSystem{exp}{\Math{\forall}\ProofSystemsFont{Exp\text{+}}\mathcal{P}}

\newcommand{\mcount}{\ell}

\newcommand{\vars}{\mathrm{\textit{vars}}}

\usepackage{tikz}
\usetikzlibrary{shapes,automata, matrix, arrows.meta, positioning}
\usetikzlibrary{decorations.pathreplacing}

\usetikzlibrary{fadings}
\tikzstyle{uedge}=[draw=blue!50!red]
\tikzstyle{fedge}=[draw=blue]
\tikzstyle{iedge}=[draw=red]
\tikzstyle{redge}=[draw=green!50!black]
\tikzstyle{rnode}=[draw,inner sep=2pt,color=black]
\tikzstyle{tnode}=[circle,minimum width=3pt,fill,inner sep=0pt]
\tikzstyle{dotnode}=[circle,minimum width=2pt,fill,inner sep=0pt]
\tikzstyle{labn}=[font=\sffamily,circle,fill=white,inner sep=1pt,draw=black]
\tikzstyle{legn}=[font=\scriptsize]
\tikzstyle{reln}=[circle,fill=white,inner sep=.4pt,draw=black]
\tikzstyle{oreln}=[circle,fill=white,inner sep=.4pt,draw=black!50,solid]
\tikzstyle{exist}=[thick, circle,fill=white,inner sep=3pt,draw=black,solid]
\tikzstyle{univ}=[thick, rectangle,fill=white,inner sep=5pt,draw=black,solid]
\tikzstyle{oree}=[thick,draw=black!50,densely dashed]
\tikzstyle{ree}=[thick,draw=black]
\tikzstyle{calcn}=[rectangle%
,rounded corners=1mm%
,draw=black,%
minimum height=.5cm,minimum width=.5cm,inner sep=5pt%
]
\tikzset{cross/.style={cross out, draw=black, fill=none, minimum size=.4cm, inner sep=0pt, outer sep=0pt}, cross/.default={1pt}}

\newtheoremstyle{TheoremNum}
    {\topsep}{\topsep}              
    {\itshape}                      
    {}                              
    {\bfseries}                     
    {}                             
    { }                             
    {\thmname{#1}\thmnote{ \textcolor{darkblue}{\bfseries #3}}}
\theoremstyle{TheoremNum}

\newtheoremstyle{TheoremNum2}
    {\topsep}{\topsep}              
    {\normalfont}                      
    {}                              
    {\itshape}                     
    {.}                             
    { }                             
    {\thmname{#1}\thmnote{ \textcolor{darkblue}{\bfseries #3}}}
\theoremstyle{TheoremNum2}

\title{On Proof Systems for \#QBF} 

\author{Sravanthi Chede}{The Institute of Mathematical Sciences (A CI of Homi Bhabha National Institute), Chennai, India \and \url{https://sravanthi-ch.github.io/webpage/} }{sravanthic@imsc.res.in}{https://orcid.org/0000-0001-7170-6156}{Supported by the ANRF J C Bose fellowship no. JCB/2023/000006.}
\author{Leroy Chew}{{Czech Technical University in Prague,  Czech Republic} \and \url{https://leroychew.wordpress.com/}}{leroy_chew@hotmail.co.uk}{https://orcid.org/0000-0003-0226-2832}{This project is supported by the European Union under the project ROBOPROX (reg. no. CZ.02.01.01/00/22\_008/0004590)
and by the Czech Science Foundation project 24-12759S.}
\author{Vaibhav Krishan}{The Institute of Mathematical Sciences (A CI of Homi Bhabha National Institute), Chennai, India\and \url{https://vaibhkrishan.github.io/} }{vaibhavk@imsc.res.in}{https://orcid.org/0009-0000-0335-1963}{}
\author{Anil Shukla}{Indian Institute of Technology Ropar, Rupnagar, India \and \url{https://anilshukl.github.io/website/} }{anilshukla@iitrpr.ac.in}{https://orcid.org/0009-0009-9051-4374}{}

\authorrunning{S. Chede, L. Chew, V. Krishan and A. Shukla} 

\Copyright{Sravanthi Chede, Leroy Chew, Vaibhav Krishan and Anil Shukla} 


\category{} 

\relatedversion{} 

\acknowledgements{We thank the anonymous SAT 2026 reviewers for their detailed and helpful suggestions.}

\nolinenumbers 

\EventEditors{Alexey Ignatiev and Stefan Szeider}
\EventNoEds{2}
\EventLongTitle{29th International Conference on Theory and Applications of Satisfiability Testing (SAT 2026)}
\EventShortTitle{SAT 2026}
\EventAcronym{SAT}
\EventYear{2026}
\EventDate{July 20--23, 2026}
\EventLocation{Lisbon, Portugal}
\EventLogo{}
\SeriesVolume{377}
\ArticleNo{10}

\begin{document}
\maketitle

\begin{abstract}

For a quantified Boolean formula (QBF), the problem of computing the number of winning strategies is known as the \#QBF problem. This problem is considered harder than the analogous \#SAT problem~\cite{Ladner89}. Recently, important proof systems for QBFs~\cite{fregeQBF, 9789811245220_0015} and \#SAT~\cite{BeyersdorffHS24,ChedeCS24} have been studied. By extending the ideas from both fields, we show that it is possible to design proof systems for \#QBF. 
Such proof systems are important not only for advancing the theory of \#QBF but also for certifying and designing better \#QBF solvers, an area that is still in its early stages~\cite{DBLP:conf/mkm/ShuklaMKS22,DBLP:journals/constraints/PlankMS24}.

In this paper, we explore \#QBF proof systems to count the number of Skolem functions. In addition to a naive system, we study \#QBF systems based on the $\forall$-expansion rule of QBFs~\cite{JanotaM15}. We observe that these systems have inherent structural weaknesses that lead to lower bounds. 
As an alternative, we propose a \#QBF proof system that we call \qmice, which consists of sound inference rules for computing and certifying the \#QBF solution, similar to the line-based \#SAT proof system \mice~\cite{fichte2022proofs,BeyersdorffHS24}. 
To demonstrate the strength of \qmice, we present various upper bounds, such as the quantified version of the propositional XOR-PAIRS formula, which is known to be hard for \mice~\cite{BeyersdorffHS24}.
Consequently, we also separate \qmice from $\forall$-expansion based \#QBF proof systems.
\end{abstract}

\section{Introduction}
\#QBF, the problem of counting the number of winning strategies of a  QBF, is challenging from several perspectives.
Firstly, from a computational complexity point of view, \#QBF is known to be \#PSPACE-complete~\cite{Ladner89}. 
\#PSPACE is at least as large as  \#P,  but potentially much larger.
Secondly, the number of winning strategies may be doubly exponential, beyond the mere exponential bounds of propositional model counting (\#SAT).
Several proof systems for the \#SAT problem have been studied. 
Notable systems include \kcps~\cite{Capelli19} and CPOG \cite{bryant2023certified} that use \emph{knowledge-compilation}.
These systems prove a model count by a transformation to a knowledge compilation class where model counting is easy.
Recently, the CLIP system~\cite{ChedeCS24} worked using a connection between circuits and PSPACE.
Another approach is to derive the model count via line-based systems, such as \mice (Model-counting Induction by Claim Extension~\cite{fichte2022proofs}).
Starting at axioms, it computes the model count using inference rules. 

In this paper, by extending ideas from QBF and \#SAT proof systems, we provide the first proof systems for the general \#QBF problem, and study its proof complexity. Such proof systems are important to advance the theory of \#QBF. 
The existence of proof systems can be useful towards \#QBF solving, in both the design and correctness of the solvers.

The main contribution of the paper is extending the \mice proof system for \#QBF (Section~\ref{Sub_sec:Q-MICE}), denoted as \qmice. \qmice is a line based \#QBF proof system, in which each line is of the form $(Q.F, A, c)$, where $Q.F$ is a true QBF, $A$ is a partial strategy that consists of Skolem functions for some existential variables, and $c$ is the number of ways to extend $A$ into complete winning strategies for the existential player of $Q.F$. \qmice consists of the axiom rule (Definition~\ref{axiom_def}), the composition rule (Definition~\ref{composition_def}) and the join rule (Definition~\ref{join_def}). Each rule can be applied if its corresponding conditions are satisfied, which are required for the rule to be sound.

In addition to proving that \qmice is  sound and complete (Theorem~\ref{thm:sound_complete_qmice}), we provide meaningful families with short proofs (Section~\ref{sec:upper_bound}) and establish an exponential separation between \qmice and the \#QBF proof system \emice (Theorem~\ref{seperation_theorem}). The \emice proof system (Section~\ref{sec:forexp_mice}) semantically expands the $\forall$ variables of the QBF (as in~\cite{JanotaM15}), then finishes the proof with a \#SAT proof system, in this case \mice.

\textbf{Related work:} 
Recently, some \#QBF solvers~\cite{DBLP:journals/constraints/PlankMS24, DBLP:conf/mkm/ShuklaMKS22,CapelliLPS24,ShawJM24} have been introduced for  solving the \#QBF problem. For example, the d4-QBF~\cite{CapelliLPS24} solver has been designed for \#QBF which uses few techniques from SAT-solving and works on the full assignment tree of the QBF recursively along with a decomposition step (\cite[Section 3.3]{CapelliLPS24}). The decomposition step identifies two or more non-connected components in the QBF matrix and separately computes the model-count in the components and merges them at the end. Another state-of-the-art \#QBF solver for the restricted case, that is, for QBFs with two quantifier alternations, is the qCounter solver~\cite{DBLP:journals/constraints/PlankMS24}. qCounter tries to enumerate winning strategies systematically one by one and then tries to certify that no more winning strategies exists.

\section{Preliminaries}\label{prelims_sec}
For a CNF $\phi$ and its variables ($\vars(\phi)$), a partial assignment $\alpha$ is a mapping from some variables $(\vars(\alpha)$ $\subseteq\vars(\phi))$ to $\{0,1\}$. 
$\phi|_\alpha$ (read as restricting $\phi$ with $\alpha$) denotes substituting values of $\vars(\alpha)$ from $\alpha$ into $\phi$. 
If $\phi|_\alpha$ is true, $\alpha$ satisfies $\phi$. 

Quantified Boolean Formulas (QBFs) are an extension of propositional Boolean formulas where each variable is quantified with one of $\{ \exists \text{ (existential)},\forall \text{ (universal)} \}$. 
In this paper, QBFs (represented as $Q.\phi$) are of the form $Q_1 X_1...Q_p X_p ~.~ \phi(X_1 \cup \dots \cup X_p)$, where $X_i$s are disjoint sets of variables; $Q_i$ $\in$ \{$ \exists$, $ \forall$\} and $Q_i\not=Q_{i+1}$, $Q$ is the quantifier prefix with $p$ alternations and the matrix $\phi$ is a CNF. 
The outermost (resp. innermost) quantified variables are $X_1$ (resp. $X_p$).
If $x\in X_i$, for any $y \in X_j$ where $j<i$, we say $y$ occurs to the left of $x$ in $Q$ (i.e. $y {\leq}_Q x$). 
$L_Q (x)$ is the set of $\forall$-variables to the left of an $\exists$-variable $x$. 

We can understand the semantics of QBFs as a game between a $\forall$ and an $\exists$ player, i.e., in the $i^{th}$ step the player corresponding to $Q_i$ assigns a Boolean value to each variable in $X_i$. At the end, the $\exists$ (resp.~$\forall$) player wins if substituting this complete assignment of variables in $\phi$ evaluates to $1$ (resp.~0).
An $\exists$ \textbf{strategy} is a set $\mathcal{S}$ containing a Boolean function $S_x$ for each $x\in \vars_\exists(Q)$, called Skolem functions.
The Skolem function for \(x\) depends only on the variables in $L_Q(x)$, i.e. \(S_x\) only takes $|L_Q(x)|$-many arguments, corresponding to the values of the universal variables in $L_Q(x)$, that would be played before $x$.

Formally, for a true QBF $Q.\phi$, existential strategy $\mathcal{S}$ is $\{S_x(L_Q(x))~|~x\in \vars_\exists(Q)\}$.
The \textbf{winning strategy} for $\exists$ is a strategy that for every possible assignment of $\vars_\forall(Q)$, the Skolem functions of the strategy respond on $\vars_\exists(Q)$ so that the QBF matrix is satisfied under the joint assignment to $\exists$ and $\forall$ variables.
A QBF is true  iff there exists a winning strategy for the $\exists$-player~\cite{AroraBarak09}. TQBF is the set of all true QBFs.
There is a dual notion of falsity (FQBF) for winning strategies (sets of Herbrand functions) for the $\forall$-player, defined in a similar manner.
Given a QBF, the problem of \#QBF is to compute the number of winning strategies for its $\exists$-player. 
We can also extend the concept of restriction to functions which respect some total order of variables.
E.g., if $A:=\{z=f(x,y), y=g(x)\}$, and the ordering of variables is $x\leq y\leq z$, $\phi|_A$ denotes $z$ substituted with $f(x,y)$, then $y$ with $g(x)$.

A QBF $Q.\phi$ can be represented as a full binary \textbf{assignment tree} where nodes in every level are labelled with one variable in the order of $Q$ (i.e., the root is the outermost variable). Two outgoing edges from every node are labelled with $0$ or $1$.
The leaves are labelled with $0/1$, which is the value that $\phi$ evaluates to when substituted with the complete assignment according to the edges on the root-to-leaf path.
The satisfiability of the QBF is then computed in a bottom-up fashion (in the assignment tree) with the syntactic meaning of the variables. 
Here, we define subtree of the assignment tree as containing the root and at least one leaf.
A \textbf{model-tree} (winning strategy) of a true QBF in this representation is a subtree that satisfies the following conditions:
\begin{itemize}
    \item it contains one outgoing edge for nodes with an $\exists$-variable and two outgoing edges for nodes with an $\forall$-variable,
    \item and all its leaves are labeled $1$s. 
\end{itemize}

Given a true QBF, we can count the number of model-trees combinatorially with a bottom-up procedure as follows.
First, consider the \(0/1\) values at the leaves as natural numbers, and proceed to their parents.
At a node with an $\exists$-variable, add the values of the children as we you only need one outgoing edge in a model-tree.
Otherwise, if the node has a $\forall$-variable, multiply the values of the children as we need both outgoing edges in a model-tree.
Finally at the root, the value equals the \#QBF answer. 
We illustrate this calculation in the following example, which we use as a running example throughout. 
\begin{example}\label{running_eg}
Let $\Phi:=\forall u_1 \exists e_1 \forall u_2 \exists e_2~.~ (e_1\lor \overline{u}_2)\land (\overline{u}_1\lor u_2\lor \overline{e}_2)\land (\overline{u}_1\lor \overline{u}_2\lor e_2)$. The complete assignment tree of $\Phi$ is shown in Figure~\ref{running_tree} for convenience.
The node labels show how the calculation proceeds as per the bottom-up procedure for calculating the \#QBF answer.
\end{example}
\begin{figure}[!h]
    \centering
\begin{tikzpicture}[scale=0.6, transform shape]

\node[univ, draw, label=above:{$4\times 1=4$}] (a) at (0,5.5) {$\bm u_1$};

\node[exist,draw, label=above:{$0+4=4$}] (b) at (-6,4.25) {$\bm e_1$};
\node[exist,draw, label=above:{$0+1=1$}] (c) at (6,4.25) {$\bm e_1$};

\node[univ, draw, label=above:{$2\times 0=0$}] (d) at (-9,2.95) {$\bm u_2$};
\node[univ, draw, label=above:{$2\times 2=4$}] (e) at (-3,2.95) {$\bm u_2$};
\node[univ, draw, label=above:{$1\times 0=0$}] (f) at (3,2.95) {$\bm u_2$};
\node[univ, draw, label=above:{$1\times 1=1$}] (g) at (9,2.95) {$\bm u_2$};

\node[exist,draw, label=93:{$1+1=2$}] (h) at (-10,1.5) {$\bm e_2$};
\node[exist,draw, label=87:{$0+0=0$}] (i) at (-8,1.5) {$\bm e_2$};
\node[exist,draw, label=95:{$1+1=2$}] (j) at (-4,1.5) {$\bm e_2$};
\node[exist,draw, label=85:{$1+1=2$}] (k) at (-2,1.5) {$\bm e_2$};
\node[exist,draw, label=95:{$1+0=1$}] (l) at (2,1.5) {$\bm e_2$};
\node[exist,draw, label=85:{$0+0=0$}] (m) at (4,1.5) {$\bm e_2$};
\node[exist,draw, label=95:{$1+0=1$}] (n) at (8,1.5) {$\bm e_2$};
\node[exist,draw, label=85:{$0+1=1$}] (o) at (10,1.5) {$\bm e_2$};

\node[] (p) at (-10.5,0) {$\bm 1$};
\node[] (q) at (-9.5,0) {$\bm 1$};
\node[] (r) at (-8.5,0) {$\bm 0$};
\node[] (s) at (-7.5,0) {$\bm 0$};
\node[] (t) at (-4.5,0) {$\bm 1$};
\node[] (u) at (-3.5,0) {$\bm 1$};
\node[] (v) at (-2.5,0) {$\bm 1$};
\node[] (w) at (-1.5,0) {$\bm 1$};
\node[] (x) at (1.5,0) {$\bm 1$};
\node[] (y) at (2.5,0) {$\bm 0$};
\node[] (z) at (3.5,0) {$\bm 0$};
\node[] (aa) at (4.5,0) {$\bm 0$};
\node[] (bb) at (7.5,0) {$\bm 1$};
\node[] (cc) at (8.5,0) {$\bm 0$};
\node[] (dd) at (9.5,0) {$\bm 0$};
\node[] (ee) at (10.5,0) {$\bm 1$};

\draw[black, ->] (a) -- (c); 

\draw[black, dashed, ->] (a) -- (b); 

\draw[black, dashed , ->] (b) -- (d); 
\draw[black, ->] (b) -- (e);
\draw[black, dashed, ->] (c) -- (f); 
\draw[black, ->] (c) -- (g);

\draw[black, dashed, dash pattern=on 2.75pt off 1.9pt, ->] (d) -- (h); 
\draw[black, ->] (d) -- (i);
\draw[black, dashed, dash pattern=on 2.75pt off 1.9pt, ->] (e) -- (j); 
\draw[black, ->] (e) -- (k);
\draw[black, dashed, dash pattern=on 2.75pt off 1.9pt, ->] (f) -- (l); 
\draw[black, ->] (f) -- (m);
\draw[black, dashed, dash pattern=on 2.75pt off 1.9pt, ->] (g) -- (n); 
\draw[black, ->] (g) -- (o);

\draw[black, dashed, dash pattern=on 2.6pt off 1.8pt, ->] (h) -- (p); 
\draw[black, ->] (h) -- (q);
\draw[black, dashed, dash pattern=on 2.6pt off 1.8pt, ->] (i) -- (r); 
\draw[black, ->] (i) -- (s);
\draw[black, dashed, dash pattern=on 2.6pt off 1.8pt, ->] (j) -- (t); 
\draw[black, ->] (j) -- (u);
\draw[black, dashed, dash pattern=on 2.6pt off 1.8pt, ->] (k) -- (v); 
\draw[black, ->] (k) -- (w);
\draw[black, dashed, dash pattern=on 2.6pt off 1.8pt, ->] (l) -- (x); 
\draw[black, ->] (l) -- (y);
\draw[black, dashed, dash pattern=on 2.6pt off 1.8pt, ->] (m) -- (z); 
\draw[black, ->] (m) -- (aa);
\draw[black, dashed, dash pattern=on 2.6pt off 1.8pt, ->] (n) -- (bb); 
\draw[black, ->] (n) -- (cc);
\draw[black, dashed, dash pattern=on 2.6pt off 1.8pt, ->] (o) -- (dd); 
\draw[black, ->] (o) -- (ee);

\end{tikzpicture}
    \caption{Assignment tree of the QBF $\Phi$ from Example~\ref{running_eg} (solid edge is $1$ and dashed edge is $0$)}
    \label{running_tree}
\end{figure}
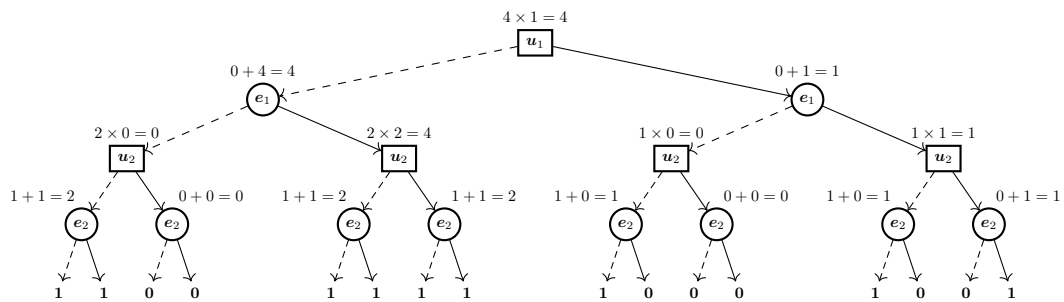

The \#QBF($\Phi$) from Figure~\ref{running_tree} is $4$, as denoted by the value at the root. Also, for instance, an explicit winning strategy from the same Figure~\ref{running_tree} is, $\mathcal{S}=\{S_{e_1}=1, S_{e_2}=u_2\}$.

\noindent \textit{Verifying that two strategies are different:}
Given two strategies $\mathcal{S}_1, \mathcal{S}_2$ of $\exists$-player for a QBF $Q.\phi$, they are different if for some complete assignment $\alpha$ to the $\forall$ variables ({witness assignment}), $\mathcal{S}_1$, $\mathcal{S}_2$ output different values for at least one of the Skolem functions corresponding to some $\exists$-variable $x$. 
That is, given two winning strategies $\mathcal{S}_1=\{S_{1,x}(L_Q(x))~|~x\in\vars_\exists(Q)\}$ and $\mathcal{S}_2=\{S_{2,x}(L_Q(x))~|~x\in\vars_\exists(Q)\}$, they are different if there exists an $x\in\vars_\exists(Q)$ and a witness assignment $\alpha$ to $\vars_\forall(Q)$ such that $S_{1,x}(\alpha)\not= S_{2,x}(\alpha)$ (here, $S_{1,x}(\alpha)$ ignores the values of the $\forall$-variables of $\alpha$ which are not in $L_Q(x)$. 
We opt for complete assignments for witnessing, as in general, the variable $x$ is not known in advance).
Given a witness assignment, it is easy to verify that the two strategies are different.

\noindent \textit{Verifying correctness of a winning strategy:}
In the QBF matrix, replacing all occurrences of $\exists$-variables with the Skolem functions of the given strategy leads to a propositional statement with only $\forall$-variables in it. A  proof that this is a tautology certifies it is a winning strategy.

We use the Cook-Reckhow~\cite{cook_1977} notion of a proof system as a sound and complete polynomial time checking function. For proof systems $g_1,g_2$ for a language $L$, $g_1$ is exponentially separated from $g_2$ is there is a family of $x\in L$ for which there are small sized $g_1$-proofs but require exponential size $g_2$-proofs.
For FQBFs (set of false QBFs),
we can provide refutations by combining propositional inference with a rule for handling universal variables. `$\forall$red' allows us to perform a $0$ or $1$ substitution on a universal variable $u$, appearing in a conjunct (i.e. any derived line). The condition is that no other variable quantified to the right of $u$ is in that particular conjunct. 

Q-Res~\cite{KBKF95} is a FQBF proof system which allows the $\forall$red rule and the propositional Resolution rule~\cite{Rob63} when the pivot is an existential variable, to construct refutations which eventually derive a $\bot$ from the input QBF.
In the FQBF-refutations we provide in this paper we will use {Frege+$\forall$red~\cite{fregeQBF} }system. 
Frege systems consist of a complete set of axioms schemas that define the Boolean connectives ($\bot, \top, \neg, \rightarrow, \leftrightarrow, \wedge, \vee$), and a Modus ponens rule. 
Frege is chosen here because it is a powerful proof system. We will omit individual propositional inference steps.

\section{Proof Systems for \#QBF}\label{naive_sec}
Naively, we can define a rudimentary proof system for the \#QBF problem.
The system enumerates all $k$ winning strategies for the given true QBF $Q.\phi$.
Then, to argue the correctness of the count, the system also proves:

\begin{itemize}
\item the correctness of each winning strategy, which requires a tautological proof of the propositional statement obtained after replacing $\exists$-variables with their respective Skolem functions in $Q.\phi$,
\item that the strategies are all distinct from each other, with explicit witness assignments for distinguishing each pair of winning strategies,
\item that no other winning strategy exists for $Q.\phi$, which we encode succinctly (in the size of $k$ and $Q.\phi$) as a false QBF $\Psi$ as follows.
\end{itemize}

\begin{claim}\label{absence_of_models_1}
Given a true QBF $Q.\phi$ and a set of winning strategies $\mathcal{S}_1,\ldots, \mathcal{S}_k$, the following QBF is false if and only if there are no other winning strategies for $Q.\phi$:
$$\Psi:=\exists \vec{\alpha}_1,...,\vec{\alpha}_k~Q~.~ \phi \land \underset{i\in [k]}{\bigwedge} \bigg(\underset{u\in U}{\bigwedge}(u\leftrightarrow {\alpha}_{i,u}) \rightarrow \underset{y\in E}{\bigvee} \Big(y\not\leftrightarrow S_{i,y}(\vec{\alpha}_i)\Big)\bigg)$$
where, $U=\vars_\forall(Q)$ and
$E=\vars_\exists(Q)$. Here, we use  $\vec{\alpha}_i$ as a witness assignment of the fact that any potential new winning strategy must differ from $\mathcal{S}_i$. $\alpha_{i,u}$ denotes the Boolean assignment corresponding to universal variable $u$ in $\vec{\alpha}_i$, that is $\alpha_i=\{\vec{\alpha}_{i,u}~|~u\in U\}$.
$S_{i,y}$ is the Skolem function of the existential variable $y$ in the winning strategy $\mathcal{S}_i$.

\end{claim}
\begin{claimproof}
    Consider any winning strategy \(\mathcal{S}\) for the $\exists$-player of \(\Psi\), which would contain an assignment for \(\vec{\alpha}_1, \ldots, \vec{\alpha}_k\).
    Now, for \(\Psi\) to be true, $\mathcal{S}$ must differ from each $\mathcal{S}_i$ with \(\vec{\alpha}_i\) as a witness.
    Moreover, the Skolem functions in $\mathcal{S}$ for each \(x \in \vars_\exists(Q)\) must also satisfy \(\phi\) on each assignment to \(\vars_\forall(Q)\).
    Hence, any winning strategy for \(\Psi\) corresponds to a \(k+1\)\textsuperscript{th} distinct winning strategy for \(Q.\phi\), and vice-versa.
    Therefore, \(\Phi\) is a false QBF if and only if there are no other winning strategies for \(Q.\phi\).
\end{claimproof}

Observe that, as opposed to the qCounter solver~\cite{DBLP:journals/constraints/PlankMS24}, our naive system is capable of certifying the correct \#QBF number of winning strategies and with no restrictions on the number of alternations on the given QBF. However, any such system has trivial lower bounds whenever the number of distinct winning strategies is (doubly) exponential.

\subsection{An Expansion-based \#QBF Proof System: \texorpdfstring{\emice}{emice} }\label{sec:forexp_mice}
We can avoid an enumerative approach to build more advanced proof systems.
The expansion-based approach~\cite{JanotaM15} for QBF solving is a well studied technique which can be applied for \#QBF.
This can be used to transform the \#QBF problem to a \#SAT problem while preserving the model count.

To be precise, given a QBF $Q.\phi$, one can expand it over all universal variable assignments so that the resulting formula consists of clauses over only existential variables. To keep track of the existential variable dependency on the universal variables, annotate the existential variable $e$ with an assignment $\alpha_e$ to $L_Q(e)$ and denote it as a new existential variable $e^{\alpha_e}$, as described in~\cite{JanotaM15}. 

Formally, the following rule is allowed to download an axiom in expansion-based systems: $$\frac{}{\{e^{\alpha_e}~|~e\in C \text{ and }e \in\vars_\exists(Q)\}\cup\{\alpha(u)~|~u\in C\text{ and } u \in\vars_\forall(Q)\}},$$ where $C\in\phi$, $\alpha$ is an assignment to all $\vars_\forall(Q)$. $\alpha_e$ is the partial assignment from $\alpha$ to variables in $L_Q(e)$ when $e\in \vars_\exists(Q)$.

For example, consider a QBF $\exists e_1 \forall u_1 \exists e_2 \forall u_2 \exists e_3.~(e_1 \vee u_1 \vee e_2 \vee \overline{u}_2 \vee e_3)\land (\overline{u_1}\lor \overline{e_2})$. Let $\alpha=\{u_1=0,u_2=1\}$ be the assignment to all universal variables of the QBF, then the corresponding clauses in the expanded CNF are $(e_1\vee 0\vee  e_2^{u_1/0}\vee 0\vee e_3^{u_1/0,u_2/1}), (1\vee \overline{e_2}^{u_1/0})$.
After expanding the QBF over all assignments to the universal variables, we obtain a SAT instance, such that its models are in bijective correspondence with the models of the QBF.
We state this formally below.

\begin{proposition}
\label{prop:expmice}
    A true QBF $\Phi$, when expanded with all complete assignments to universal  variables, leads to a CNF formula $\phi$ such that every model of $\phi$ leads to a distinct model of $\Phi$ and vice versa. In other words, their model-counts are the same.
\end{proposition}

Before proving the bijection of Proposition~\ref{prop:expmice} for general QBFs, we would like to discuss a simple case below for the sake of intuition. 
\begin{example}
Consider a QBF $\Phi$ with one universal variable u, existential variables $x_1,..., x_n$, and its CNF matrix `$\chi$’. Assume that each $x_i$ depends on $u$. After expansion, the CNF obtained is $\phi = \chi(u=0, x_1^{u/0},..., x_n^{u/0}) \wedge  \chi(u=1, x_1^{u/1},..., x_n^{u/1}))$.

Any model of $\phi$ assigns a value to each ($x_i^{u/0}=a_{i, 0}$) and ($x_i^{u/1} = a_{i, 1}$), such that $\phi$ is satisfied. The corresponding model of the QBF $\Phi$ chooses ($x_i = a_{i, b}$) when $u = b$.
Given two different models for $\phi$, they must differ at some ($x_i^{u/b}$). Then, the corresponding winning strategies also differ, as the Skolem function for $x_i$ will differ at $u=b$.

Conversely, consider a winning strategy $\mathcal{S}$ for $\Phi$ with $\{S_{x_i}(u)\}$ as the individual Skolem functions. The corresponding satisfying assignment for $\phi$ assigns $x_i^{u/b} = S_{x_i}(b)$. Again, two different winning strategies lead to different satisfying assignments. Hence, we establish the required bijection for $\Phi$.
\end{example}
\begin{proof}[Proof of Proposition~\ref{prop:expmice}]
    Consider a QBF formula $\Phi=Q.\chi$ with the universal variables \(u_1,...,u_m\), existential variables \(x_1, \ldots, x_n\), and the CNF matrix \(\chi(u_1,\ldots,u_m, x_1, \ldots, x_n)\).
After complete semantic expansion of $\Phi$, the SAT instance we obtain is: 

$\phi:=\underset{\alpha\in\langle u_1,...,u_m\rangle}{\bigwedge} \chi\Big(u_1...u_m=\alpha, \big\{x_i^{L_Q(x_i)/\alpha[L_Q(x_i)]}~|~\text{ for }i\in[n]\big\}\Big)$

where $\langle U\rangle$ is the set of all complete assignments to variables in $U$, $L_Q(x)$ is the set of all universal variables occurring to the left of the existential variable $x$ in $Q$. $\alpha[L_Q(x)]$ is the set of individual $0/1$-assignments from $\alpha$ for variables in $L_Q(x)$. For example, if $\alpha=\{u_1=0,u_2=1,u_3=0\}$ then $\alpha[u_2,u_3]=[1,0]$ and $\{u_2,u_3\}/\alpha[u_2,u_3]=\{u_2=1,u_3=0\}$.

Now, any assignment to this SAT formula is of the form \(x_i^{L_Q(x_i)/\alpha[L_Q(x_i)]} = a_{i, \alpha[L_Q(x_i)]}\), where \(a_{i, \alpha[L_Q(x_i)]} \in \{0, 1\}, \alpha \in \{0^m,..., 1^m\}\) and \(i \in [n]\).
Below we show that any satisfying assignment $(x_i^{L_Q(x_i)/\alpha[L_Q(x_i)]} = a_{i, \alpha[L_Q(x_i)]})$ corresponds to a winning strategy of the QBF $\Phi$ where, the existential player chooses \(x_i = a_{i, \alpha[L_Q(x_i)]}\) if the universal chooses \(u = \alpha\) and vice-versa. 

One direction is easy: when a satisfying assignment for $\phi$ is known the corresponding winning strategy is simply to choose $x_i$ to be $a_{i,\alpha[L_Q(x_i)]}$. 
In the other direction: when a winning strategy $\mathcal{S}$ for $\Phi$ (i.e. $\{S_{x_i}(L_Q(x_i))\}$) is known the corresponding satisfying assignment would be to set $x_i^{L_Q(x_i)/\alpha[L_Q(x_i)]}=S_{x_i}(\alpha[L_Q(x_i)])$.

This is clearly a bijective correspondence, as any two 
distinct assignments, would differ in at least one $x_i^{L_Q(x_i)/\alpha[L_Q(x_i)]}$ value and the corresponding strategies would need to be different as $S_{x_i}(\alpha[L_Q(x_i)])$ is just a Boolean value and to represent both $0$ and $1$ (from both assignments), would require two different $S_{x_i}$ functions and vice-versa.
\end{proof}

In \emice, one needs to expand the QBF with all possible assignments to universal variables.
Otherwise, if one expands on a smaller set of assignments to obtain a CNF $\phi'\subset \phi$, one can only guarantee \#QBF$(\Phi)\leq \#\text{SAT}(\phi')$.  
This is because, for any CNF $\phi'$ obtained after dropping a few clauses from a CNF $\phi$, we have \#SAT$(\phi')\geq$\#SAT$(\phi)$.

\begin{corollary}\label{cor_exp}
    For any \#SAT proof system $\mathcal{P}$, the proofs in the \#QBF proof system \exp for a QBF $\Phi$ need exponential size in terms of $|\vars_\forall(\Phi)|$.
\end{corollary}

\subsection{A Line-based \#QBF Proof System: \qmice}\label{Sub_sec:Q-MICE}
In \emice, eliminating all universal quantifiers at the beginning of the proof can cause an exponential explosion, even where the model count may otherwise be straightforward. In this section, we introduce a line-based \#QBF proof system \qmice.
The idea is that interleaving quantifier handling and other rules allows better control of the proof size in instances where the model count ought to be easy to determine. 
\qmice is inspired from \mice for \#SAT ~\cite{fichte2022proofs}.
A claim in \mice is of the form $(F, A, c)$ where $F$ is a CNF formula, $A$ is a partial assignment, and $c$ is the correct number of satisfying assignments to $F|_A$. Below, we define the \qmice system for \#QBF along similar lines, borrowing notations from~\cite{BeyersdorffHS24} for ease.

\subsubsection{Inference rules for \qmice}
In  \qmice, we use the claims of the following format: $(Q.F, A, c)$, where $Q$ is the quantifier prefix.
$F$ is a CNF formula (matrix), $c$ is a non-negative integer written in decimal\footnote{Future improvements could use succinct integer representations to avoid trivial exponential lower bounds.} and $A$ is a partial strategy.
In a partial strategy \(A\), some of the $\exists$-variables $x$ are allotted strategies in the form of a Skolem function $S_x(L_Q(x))$. The set of variables which has allotted strategies in \(A\) is represented by $\vars(A)$. One representation for the Skolem functions in $A$ is via Boolean circuits, one for each $x\in \vars(A)$. In this representation, \(A\) becomes a collection of circuits.

A claim $(Q.F, A, c)$ is said to be a \textbf{valid claim} if $c$ is the correct number of ways to extend $A$ to get a winning strategy $A'$ (by assigning strategies to the remaining $\exists$-variables) for the $\exists$-player in $Q.F$.
The \textbf{size} of a claim is the number of symbols needed to represent the claim in some fixed representation. For example, circuits in $A$ can be represented as CNFs using Tseitin transformation~\cite{Tse70}.
We need the following definitions for \qmice.

\begin{definition}[Restriction of $\forall$-variables in a QBF]\label{restriction_qbf}

Let $Q.F$ be a QBF. 
By a restriction of the $\forall$-variable ($y =b\in\{0,1\}$) in a QBF $Q.F$ (denoted as $Q.F|_{\{y=b\}}=Q'.F\land (y=b)=Q'.F'$), we mean that the quantification of $y$ is flipped to $\exists$ in the prefix $Q'$ and an additional unit clause forcing the restriction $y=b$ (i.e. if $b=1$ then $(y)$ else $(\overline{y})$) is added in the matrix. All other clauses in $F$ remain. 
\end{definition}

We start the description of the \qmice system by first defining its inference steps and providing examples to illustrate their usage.

\begin{definition}[Axiom rule] \label{axiom_def}
    For a QBF $Q.F$, one can derive: 
    {\fontsize{14.5}{17.4}\selectfont $\frac{}{(Q.F|_\rho,~A,~c)}$} 
    where
    \begin{description}
        \item [(A-1)] $\rho$ is a (partial) assignment to $\vars_\forall(Q)$ and $Q.F|_\rho$ is denoted as $Q'.F'$ (as in Def~\ref{restriction_qbf}).
        \item [(A-2)] $F|_{\{\rho\cup A\}}$ is a tautology and an explicit proof of this is provided.
        \item [(A-3)] $\vars(A)\subseteq \vars_\exists(Q)$, calculate $c$ as follows:
        \begin{itemize}
            \item Initially, $\mcount=1$. Starting from the innermost quantified variable in $Q$ (say $x$):
            \item if $x\in\vars_\exists(Q) ~\text{and}~ x\notin \vars(A)$ then: $\mcount \gets 2\cdot\mcount $ (i.e. double the count),
            \item if $x\in\vars_\forall(Q) ~\text{and}~ x\not\in\vars(\rho)$ then: $\mcount \gets\mcount^ 2$ (i.e. square the count),
            \item at the end of the reverse quantification sequence, $c\gets \mcount$.
        \end{itemize}
    \end{description}
\end{definition}

The intuition of Axiom rule is as follows: if $\phi$ be a tautology, then for any quantification $Q$ to $\vars(\phi)$, $Q.\phi$ will be a true QBF. In fact, all strategies are winning in $Q.\phi$, allowing us to easily count the number of model.
For instance, in the QBF from Example~\ref{running_eg}, we can derive the following axioms: 
$\big(\Phi|_{\{u_1=0\}}, \{e_1=1\},4\big)$, $\big(\Phi|_{\{u_2=0\}}, \{e_1=\overline{u}_1, e_2=0\},1\big)$, $\big(\Phi|_{\{u_1=1, u_2=0\}}, \{e_1=1, e_2=0\},1\big)$. Tautology proofs here are trivial. 

\begin{lemma}\label{axiom_sound}
    Any claim $(Q.F|_\rho,~ A,~ c)$ derived by using the Axiom rule of Definition~\ref{axiom_def} is valid i.e., $c$ is the correct number of models for the QBF $Q.F|_\rho \land {\bigwedge}_{x\in \vars(A)}(x\leftrightarrow A(x))$.
\end{lemma}
\begin{proof}
In this rule, we find a subtree ($Q'.F'\cup A$) of the assignment tree ($Q.F$) which has all $1$-leaves and compute all possible valid models in this subtree. This computation is just combinatorial (described in Section~\ref{prelims_sec}). For the correctness of this model-count, we argue inductively on the subtree ($Q'.F'\cup A$):  
Base case: the model-count at the leaves is $1$. Inductively, model-count at unrestricted existential nodes $x$ i.e. $x \not\in\vars(A)$ (or universal nodes $x$ i.e. $x \not\in\vars(\rho)$) is the addition (or multiplication) of model-counts at their two children owing to the structure of any model-tree. 
The remaining two cases of induction are restricted existential variables $x$ (i.e. $x \in\vars(A)$) or universal variables $x$ (i.e. $x \in\vars(\rho)$) of $Q$. The model-count in these cases is retained from the only remaining child as both are now existentially quantified in $Q'$ and need the only remaining outgoing edge in any model-tree.
\end{proof}

For a QBF $Q.F$, one can drop some variables and their strategies from the partial strategy part of the claims using one of the following composition rules. Note that all these rules require a proof for the absence of models statement, which is defined at the end.

\begin{definition}[Composition rules]\label{composition_def}
    Given a QBF $Q.F$,
\begin{enumerate}
\vspace{0.1cm}
    \item (Composition-a:)
    \hspace*{1cm}{\fontsize{14.5}{17.4}\selectfont$\frac{(Q.F|_\rho,~ A_1,~ 1),...,(Q.F|_\rho,~ A_n,~ 1)}{(Q.F|_\rho,~ A,~ n)}$}, where
    \begin{description}
\vspace{0.05cm}
        \item [(C-a1)] for all $i\in[n], \vars(A_i) = \vars_\exists(Q)$ and $\vars(\rho)\subseteq \vars_\forall(Q)$,
         \item [(C-a2)] for $i\not=j\in [n]$, $A_i\not= A_j$ (providing assignments witnessing that $A_i, A_j$ disagree),
         \item [(C-a3)] $A\subseteq \cap_{i\in [n]} A_i$ (a subset of strategies where all $A_i$s are syntactically equal).
    \end{description}
\vspace{0.1cm}
\item (Composition-b:)
    \hspace*{1cm} {\fontsize{14.5}{17.4}\selectfont$\frac{(Q.F|_\rho,~ A_1,~c_1)}{(Q.F|_\rho,~ A,~ c_1)}$}, where
    \begin{description}
\vspace{0.05cm}
        \item [(C-b1)] $A\subseteq A_1\subseteq \vars_\exists(Q)$ and $\vars(\rho)\subseteq \vars_\forall(Q)$. 
    \end{description}
\vspace{0.1cm}
\item (Composition-c:)
    \hspace*{1cm} {\fontsize{14.5}{17.4}\selectfont$\frac{(Q.F|_\rho,~ A_1,~c_1),...,(Q.F|_\rho,~ A_n,~ c_n)}{(Q.F|_\rho,~ A,~ \Sigma_{i\in n} c_i)}$}, where 
    \begin{description}
\vspace{0.05cm}
        \item [(C-c1)] $\vars(A_1)=...=\vars(A_n)\subseteq \vars_\exists(Q)$
        and $\vars(\rho)\subseteq \vars_\forall(Q)$,
        \item [(C-c2)] for $i\not=j\in [n]$, $A_i\not= A_j$ and $A\subseteq \cap_{i\in [n]} A_i$.
        \item [(C-c3)] Let $x$ be the innermost $\exists$-variable in $Q$ such that $x \in  \vars(A_i)\setminus \vars(A)$.
        \begin{itemize}
            \item For all $y\in \vars_\exists(Q)$ and $y\leq_Q x$ $\implies y\in\vars(A_i)$.
            \item For all $y\in \vars_\forall(Q)$ and $y\leq_Q x$ $\implies y\in \vars(\rho)$.
        \end{itemize}
    \end{description}
\end{enumerate}
For ease, we denote $Q.F|_\rho$ as $Q'.F'$ in these rules.
All these rules additionally need an FQBF-proof of the \textbf{absence of models statement}: $\Psi:= Q'.F'\cup A \cup \{\overline{A_i\setminus A}\}_{i\in n}$ where $n$ is the number of hypothesis claims.
$\Psi$ encodes as a QBF the negation of the fact that there exists no other winning strategies for the $\exists$-variables in $\vars(A_i)\setminus \vars(A)$. 
The exact encoding of this statement is provided in the next claim.
\end{definition}

\begin{claim}\label{absence_of_models}
$\Psi$ in the composition rules (Definition~\ref{composition_def}) can be encoded as a short QBF.
\end{claim}
\begin{claimproof}
    This encoding is a generalization of that defined in Claim~\ref{absence_of_models_1}.
    For an input QBF $Q.F$, the composition rule in general is of the form where there are $n$ partial strategies $A_1,...,A_n$ over the same set of existential variables $X\subseteq \vars_\exists(Q)$ in the hypothesis, and a subset of these strategies ($X'\in X$) which appear in all of them are retained as $A$ in the conclusion.
    The `absence of models' statement encodes that there are no other winning strategies (in $Q'.F'$) for the variables in $X\setminus X'$ when also adhering to strategies from $A$.
    
    Recall that $A_i$s are a set consisting of one function $S_{A_i,x}$ for every $x\in \vars(A_i)$. So, in $Q'.F'$, enforcing the strategies in $A$ is simply adding clauses $(x\leftrightarrow S_{A,x})$ for every $x\in X'$.

    Now to see if there is any $n+1^{\text{th}}$ winning strategy for variables $\in X\setminus X'$, we need it to be different than those already in the hypotheses. As discussed in Section~\ref{prelims_sec}, for two strategies to be considered different, there must be some witness assignment ($\alpha$ below) to $\forall$-variables of $Q'$ (say $U'$) such that the value of the function computed by these strategies is different. Now, if the following QBF encoding is true, it implies that there indeed is a $n+1^{\text{th}}$ strategy which is missing in the hypothesis and the rule cannot be used.
    $$\Psi:=\exists \vec{\alpha}_1,...,\vec{\alpha}_n~Q'. F' \land \underset{x\in X'}{\bigwedge} \big(x\leftrightarrow S_{A,x}\big) \land \underset{i\in [n]}{\bigwedge} \bigg(\underset{u\in U'}{\bigwedge}(u\leftrightarrow {\alpha}_{i,u}) \rightarrow \underset{y\in X\setminus X'}{\bigvee} \Big(y\not\leftrightarrow S_{A_i,y}(\vec{\alpha}_i)\Big)\bigg)$$

    However, if the above encoding is a false QBF, it certifies that there is no such missed strategy and we can proceed to drop the strategy restrictions on variables $\in X\setminus X'$. 
\end{claimproof}

The intuition for the composition rules: Composition-a is useful to directly count complete winning strategies. Composition-b is useful when some variables have unique winning Skolem functions in the subtree. Composition-c is useful for combining subtrees which all agree on the restrictions of universal variables and individually are a path graph from the root to the innermost variable where the hypotheses strategies differ.

For instance, in the QBF from Example~\ref{running_eg}, $\Phi:=\forall u_1 \exists e_1 \forall u_2 \exists e_2~.~ (e_1\lor \overline{u}_2)\land (\overline{u}_1\lor u_2\lor \overline{e}_2)\land (\overline{u}_1\lor \overline{u}_2\lor e_2)$ we can derive the following claims.
Using composition-b, we can derive:
\[\frac{(\Phi|_{u_1=1},~ \{e_1=1,e_2=u_2\},~1)}{(\Phi|_{u_1=1},~ \{e_1=1\},~1)}\]
The absence of models statement 
$\Psi$ is:\[\exists \alpha_1 \exists u_1 \exists e_1 \forall u_2 \exists e_2~.~(e_1\lor\overline{u}_2)\lor (\overline{u}_1\lor u_2\lor\overline{e}_2)\land (\overline{u}_1\lor\overline{u}_2\lor e_2)\land (u_1)\land (e_1)
\land \big((u_2=\alpha_1)\rightarrow (e_2\not= \alpha_1) \big)\]
For an FQBF-proof of $\Psi$ in Frege+$\forall$red~\cite{fregeQBF}, we use propositional inferences to derive $(u_2\lor\alpha_1)$ and $(\overline{u}_2\lor\overline{\alpha}_1)$. Using universal reduction to drop variable $u_2$, we can derive {$(\alpha_1)$} and {$(\overline{\alpha}_1)$}, which gives the needed contradiction.\\
Using composition-a (or -c) we can derive the following:
\[\frac{(\Phi|_{u_1=0,u_2=0},~\{e_1=0, e_2=0\},~1) ~,~ (\Phi|_{u_1=0,u_2=0},~\{e_1=0, e_2=1\},~1)}{(\Phi|_{u_1=0,u_2=0},~\{e_1=0\},~2)}\]
Here, the corresponding absence of models statement is trivially false.

\begin{lemma}\label{composition_sound}
    Any claim $(Q'.F',~ A,~ c')$ derived by using the composition rules of Definition~\ref{composition_def} from valid hypothesis claims is also valid.
\end{lemma}
\begin{proof}

In composition-a, every hypothesis claim considers a full strategy for the existential player. That is, each hypothesis is a distinct model for the subtree corresponding to $Q'.F' \cup A$ and no other models exist for the same (due to absence of models proof). Therefore, adding them up is the correct number of models for the formula $Q'.F' \cup A$.  

In composition-b, there is only one hypothesis strategy for variables in $\vars(A_1)\setminus\vars(A)$ and no other winning strategies are possible due to the absence of models proof. 
So, there is only one useful branch from these $\exists$-variables in the subtree of the assignment tree, and these branches need to be in any model tree. Therefore, the count at the root remains the same. 

In composition-c, $x$ is the innermost variable from $\vars(A_i)\setminus\vars(A)$ and every variable to the left of $x$ either has some strategy in every $A_i$ or is restricted with $0/1$ in the QBF. Hence, the corresponding subtree of the assignment tree for these hypotheses is a line graph from root to $x$. Let $y$ be the outermost variable from $\vars(A_i)\setminus\vars(A)$. In the subtree corresponding to $Q'.F'\cup A$, it is a line graph up to $y$ and it only branches on every $\exists$-variable in $\vars(A_i)\setminus\vars(A)$ up to $x$ and at $x$ it has one branch with count $c_i$ and the other with count $0$ (due to the absence of models proof). Now the total count at $y$-node is just combinatorial and equal to $\Sigma_{i\in n}c_i$. As it is just a line graph until the root, this value propagates as is.  
\end{proof}

For a QBF $Q.F$, one can drop the $\forall$-variable restrictions  using the following join rule.

\begin{definition}[Join rule] \label{join_def}
    For a QBF $Q.F$, one can derive
    \[\frac{(Q.F|_{\rho\cup\{y=0\}},~ A,~ c_1),(Q.F|_{\rho\cup\{y=1\}},~ A,~ c_2)}{(Q.F|_{\rho},~ A,~ c_1\cdot c_2)}\]
    when the following conditions are satisfied:
    \begin{description}
        \item [(J-1)] $\rho$ should be a (partial) assignment to $\vars_\forall(Q)$ and $\vars(A)\subseteq \vars_\exists(Q)$.
        
        For ease, the quantifier in both antecedents is $Q'$ and $Q.F|_{\rho}$ is denoted as $Q''.F_3$.
        \item [(J-2)] $y \in \vars_\forall(Q)$, $y \in \vars_\exists(Q')$, $y \in \vars_\forall(Q'')$ and unit clauses ($y$), ($\overline{y}$) are dropped. 
        \item [(J-3)] For all $x\in \vars_\exists(Q)$ and $x\leq_Q y$ $\implies x\in\vars(A)$.
        \item [(J-4)] For all $x\in \vars_\forall(Q)$ and $x\leq_Q y$ $\implies x\in \vars(\rho)$.
    \end{description}
\end{definition}

The intuition behind the join rule is to combine two subtrees (say $G_1,G_2$) agreeing everywhere except for one restriction of a  $\forall$-variable ($y$) and $G_1,G_2$ are simple path graphs from the root to $y$.  
In
Example~\ref{running_eg}, $\Phi:=\forall u_1 \exists e_1 \forall u_2 \exists e_2~.~ (e_1\lor \overline{u}_2)\land (\overline{u}_1\lor u_2\lor \overline{e}_2)\land (\overline{u}_1\lor \overline{u}_2\lor e_2)$, using join we can derive:
\[\frac{(\Phi|_{u_1=0},~\{e_1=1\},~4) ~,~ (\Phi|_{u_1=1},~\{e_1=1\},~1)}{(\Phi,~\{e_1=1\},~4)}\]

\begin{lemma}\label{join_sound}
    Any claim $(Q.F|_\rho,~ A,~ c_1\cdot c_2)$ derived by using the join rule of Definition~\ref{join_def} from valid hypothesis claims is also valid.
\end{lemma}
\begin{proof}
In the hypothesis, $y$ is a $\forall$-variable of $Q.F$ which is restricted and is now a $\exists$-variable in $Q'$. Every variable to the left of $y$ either has some strategy in $A$ or is restricted with $0/1$ in $\rho$ (due to J-3, J-4). Hence, the corresponding subtree of the assignment tree for these hypotheses is a path graph from root to $y$. The counts at root of hypotheses are the same at the $y$-node. Now, the conclusion subtree ($Q''.F_3\cup A$) is similarly a path graph from root up to $y$ and at $y$ branches with $0$, $1$. Now, any model of this subtree should have both branches in its model-tree, therefore $c_1\cdot c_2$ is the correct number of models in the subtree.
\end{proof}

\noindent We now define the \#QBF proof system \qmice using the inference rules introduced above.

\begin{definition}[\qmice]
A \qmice proof of a QBF $Q.F$ with $k$ winning strategies for the existential player is a sequence of lines $L_1,...,L_s$ where $L_s$ contains the claim $(Q.F,~ \emptyset,~ k)$ and every line $L_i$ for $i\in[s]$ is derived by one of the following rules:
\begin{enumerate}
    \item Axiom rule from Definition~\ref{axiom_def} along with the required tautology proof.
    \item One of the Composition rules from Definition~\ref{composition_def} with hypothesis from $L_1,...,L_{i-1}$, along with witnesses separating the hypothesis (partial) strategies, and absence of models proof
    \item Join rule from Definition~\ref{join_def} with hypothesis from $L_1,...,L_{i-1}$.
\end{enumerate}
The `length' of a \qmice proof is the number of lines in it, and its `size' is the size of all claims along with additional witnesses, FQBF-proofs and tautology proof sizes.
\end{definition}

A complete \qmice proof of $\Phi$ from Example~\ref{running_eg} is as follows:

\begin{example}
The QBF from Example~\ref{running_eg}, $\Phi:=\forall u_1 \exists e_1 \forall u_2 \exists e_2~.~ (e_1\lor \overline{u}_2)\land (\overline{u}_1\lor u_2\lor \overline{e}_2)\land (\overline{u}_1\lor \overline{u}_2\lor e_2)$, has a \qmice proof as follows:\\
$L_1:= \big(\Phi|_{\{u_1=0\}}, \{e_1=1\},4\big)$ (Axiom, tautology proof is easy)\\
$L_2:= \big(\Phi|_{\{u_1=1\}}, \{e_1=1, e_2=u_2\},1\big)$ (Axiom, tautology proof is easy)\\ 
$L_3:= (\Phi|_{\{u_1=1\}}, \{e_1=1\},1)$ (Composition-b on $L_2$) 
\begin{itemize}
    \item The absence of models statement is: 
$\Psi_1:=\exists \alpha_2 \exists u_1 \exists e_1 \forall u_2 \exists e_2~.~(e_1\lor\overline{u}_2)\lor (\overline{u}_1\lor u_2\lor\overline{e}_2)\land (\overline{u}_1\lor\overline{u}_2\lor e_2)\land (u_1)\land (e_1)\land \big((u_2=\alpha_2)\rightarrow (e_2\not= \alpha_2) \big)$.
\item FQBF-proof of $\Psi_1$ in Q-Res~\cite{KBKF95} is as follows after unit propagation steps:
{
\begin{prooftree}
    \AxiomC{$(u_2\lor\overline{e}_2)$}
    \AxiomC{$(u_2\lor\alpha_2\lor e_2)$}
    \LeftLabel{(Res)}
    \BinaryInfC{$(u_2\lor\alpha_2)$} \LeftLabel{($\forall$red)}
    \UnaryInfC{$(\alpha_2)$}
    \AxiomC{$(\overline{u}_2\lor e_2)$}
    \AxiomC{$(\overline{u}_2\lor\overline{\alpha}_2\lor \overline{e}_2)$}
    \RightLabel{(Res)}
    \BinaryInfC{$(\overline{u}_2\lor\overline{\alpha}_2)$}
    \RightLabel{($\forall$red)}
    \UnaryInfC{$(\overline{\alpha}_2)$}\
    \RightLabel{(Res)}
    \BinaryInfC{$\bot$}
\end{prooftree}
}
\end{itemize}
$L_4:= (\Phi,\{e_1=1\},4)$ (Join on $L_1, L_3$)\\
$L_5:= (\Phi,\emptyset,4)$ (Composition-b on $L_4$)
\begin{itemize}
    \item The absence of models statement is: $\Psi_2:=\exists \alpha_1 \forall u_1 \exists e_1 \forall u_2 \exists e_2~.~ (e_1\lor\overline{u}_2)\land (\overline{u}_1\lor u_2\lor \overline{e}_2)\land (\overline{u}_1\lor \overline{u}_2\lor e_2)\land \big((u_1\not=\alpha_1)\rightarrow (e_1\not=1)\big)$
    \item FQBF-proof of $\Psi_2$ in Q-Res~\cite{KBKF95} is as follows:
    {
    \vspace{-0.25cm}
    \begin{prooftree}
\def\defaultHypSeparation{\hskip .08in}
\AxiomC{$(u_1\lor\alpha_1\lor\overline{e}_1)$}
        \AxiomC{$(e_1\lor \overline{u}_2)$}\RightLabel{~($\forall$red)}
        \UnaryInfC{$(e_1)$} \LeftLabel{~(Res)}
        \BinaryInfC{$(u_1\lor\alpha_1)$} \RightLabel{($\forall$red)}
        \UnaryInfC{$(\alpha_1)$}

        \AxiomC{$(e_1\lor \overline{u}_2)$}
        \UnaryInfC{$(e_1)$}
        \AxiomC{$(\overline{u}_1\lor\overline{\alpha}_1\lor\overline{e}_1)$}
        \RightLabel{(Res)}
        \BinaryInfC{$(\overline{u}_1\lor\overline{\alpha}_1)$} \RightLabel{($\forall$red)}
        \UnaryInfC{$(\overline{\alpha}_1)$} \LeftLabel{~(Res)}
        \BinaryInfC{$\bot$}
    \end{prooftree}
    }
\end{itemize}
\end{example}

\begin{theorem}\label{thm:sound_complete_qmice}
\qmice is a polynomial-time verifiable sound and complete \#QBF system.
\end{theorem}
\begin{proof}
Since all the inference rules are sound (Lemmas~\ref{axiom_sound},\ref{composition_def},\ref{join_sound}), the \qmice proof system is sound. 
For completeness: Use the Axiom rule for every winning strategy of $Q.F$ with $c=1$. Using Composition-a rule with $A=\emptyset$, add these axiom claims to derive the correct \#QBF answer: 
If you consider all winning strategies in the absence of models statement, it is indeed a false QBF and any complete FQBF-proof system suffices. Also every rule of \qmice is easily verifiable, hence any \qmice-proof is verifiable in time polynomial in the proof-size.
\end{proof}

Note that if the input QBF is false, one can use any one of the composition rules with no hypothesis strategies to derive the answer that \#QBF answer is $0$. However, this is equivalent to directly using a FQBF proof system to to prove that the input QBF is false.

\section{Upper bounds}\label{sec:upper_bound}
In this section, we show the strength of \qmice by proving upper bounds for two example QBF families that have roots in theory~\cite{BeyersdorffHS24} and practice~\cite{IhsanD23}, respectively.
\subsection{A QBF family based on \textsf{XOR-PAIRS}}\label{sec:xor_upperbound}

$$\mathcal{Q}.\oplus_n:=\forall u_1, ..., u_n~\underset{i\not=j\in[n]}{\exists} x_{i,j}~.~ {\bigwedge}_{i\not=j\in[n]}(x_{i,j}\leftrightarrow u_i\oplus u_j)$$ 

$\mathcal{Q}.\oplus_n$ has $n^2$ variables out of which 
$n$ are universal variables, hence it requires super-polynomial ($\Omega(2^n)$) sized proofs in the \emice system (Corollary~\ref{cor_exp}). 

\noindent Below, we show that $\mathcal{Q}.\oplus_n$ has a constant length and polynomial size \qmice proof $\pi=\{L_1,L_2\}$, where $L_1$ is derived using the Axiom rule consisting of the entire QBF $\mathcal{Q}.\oplus_n$, a complete winning strategy $A = \{S_{x_{i,j}}= u_i\oplus u_j | i\not=j\in[n]\}$ and the count $1$, accompanied by a tautology proof of $\oplus_n|_A$. Line $L_2$ is derived by applying the Composition-b rule on $L_1$ by dropping the strategy restrictions in $A$ accompanied by a FQBF proof of the absence of models statement. This maintains the previous count, hence \#QBF$(\mathcal{Q}.\oplus_n) = 1$. Precisely:

(Axiom rule) $L_1:=\big(\mathcal{Q}.\oplus_n,~ A:=\{S_{x_{i,j}}= u_i\oplus u_j | i\not=j\in[n]\},~1\big);$ 
the tautology proof of $\oplus_n|_A$ is a straightforward case of any introduction rule of $\leftrightarrow$ in any Frege system.

\noindent (Composition-b rule) $L_2:=\big(\mathcal{Q}.\oplus_n,~ \emptyset,~ 1\big);$ 
the absence of models statement as per Claim~\ref{absence_of_models} is:

$\Psi:= \exists \vec{\alpha}_{1}~\forall u_1, ..., u_n~{\exists}_{i\not=j\in[n]} x_{i,j}~.~{\bigwedge}_{i\not=j\in[n]}(x_{i,j}\leftrightarrow u_i\oplus u_j)\\\hspace*{1cm}\land \big((u_1=\alpha_{1,u_1}\land ...\land u_n=\alpha_{1,u_n})\rightarrow ({\bigvee}_{i\not=j\in[n]} x_{i,j}\not=\alpha_{1,u_i}\oplus \alpha_{1,u_j})\big)$

\noindent FQBF proof of $\Psi$ using Frege+$\forall$red : (for simplicity, we denote $\alpha_{1,u_i}$ as $\alpha_i$)
{\small
\begin{prooftree}
\AxiomC{$x_{i,j}\leftrightarrow u_i\oplus u_j$}
\AxiomC{$(u_1=\alpha_1\land ...\land u_n=\alpha_n)\rightarrow (\underset{i\not=j\in[n]}{\bigvee} x_{i,j}\not=\alpha_i\oplus \alpha_j)$} \RightLabel{(Prop. Frege inference)} 
\BinaryInfC{$(u_1=\alpha_1\land ...\land u_n=\alpha_n)\rightarrow (\underset{i\not=j\in[n]}{\bigvee} u_i\oplus u_j\not=\alpha_i\oplus \alpha_j)$} \RightLabel{($\forall$red with $u_i=\alpha_i$ for $i\in[n]$)}
\UnaryInfC{$(\alpha_1=\alpha_1\land ...\land \alpha_n=\alpha_n)\rightarrow (\underset{i\not=j\in[n]}{\bigvee} \alpha_i\oplus \alpha_j\not=\alpha_i\oplus \alpha_j)$} \RightLabel{(Prop. Frege inference)}
\UnaryInfC{$\bot$}
\end{prooftree}
} 
\subsection{An Indexed Affine QBF Family}

$$\Gamma=\mathcal{Q}.\chi_n:= \forall u_1 \ldots u_{\log n} ~\exists x ~\forall v_1 \ldots v_n ~\exists y_1 \ldots y_n~.~
{\bigwedge}_{i \neq (u)} (y_i = v_i \oplus x) \wedge (y_u = v_u \oplus \overline{x})$$

\noindent The QBF $\Gamma$ can be thought of as a simple encryption process, where the function for \(x\) w.r.t. the index `\(u\)' is the encryption scheme, which is applied to the plain-text `\(v\)' to obtain `\(y\)' as the cipher-text.
As the index \(u\) for \(y_u\) and \(v_u\) is denoted by its binary expansion represented by \(u_1, \ldots, u_{\log(n)}\), it's not immediately clear how to encode the matrix as a CNF.
However, $u$ has only $\log n$ bits so one way to represent $y_u, v_u$ is to go through each assignment $u \in \{0,1\}^{\log n}$ in some order.

Following this idea, given an assignment to \(a \in \{0, 1\}^{\log(n)}\), we create four clauses to represent the condition \((u_1 = a_1, \ldots, u_{\log(n)} = a_{\log(n)}) \rightarrow (v_a = y_a \oplus \overline{x})\). 
For example, if \(a_1 = \ldots = a_{\log(n)} = 0\), the four clauses are:
\begin{eqnarray*}
    u_1 \vee \ldots \vee u_{\log(n)} \vee v_0 \vee y_0 \vee x &~~~ u_1 \vee \ldots \vee u_{\log(n)} \vee v_0 \vee \overline{y}_0 \vee \overline{x} \\
    u_1 \vee \ldots \vee u_{\log(n)} \vee \overline{v}_0 \vee y_0 \vee \overline{x} &~~~ u_1 \vee \ldots \vee u_{\log(n)} \vee \overline{v}_0 \vee \overline{y}_0 \vee x
\end{eqnarray*}
Similarly, for each \(a' \neq a\), we create four clauses to represent the condition \((u_1, \ldots, u_{\log(n)}) = (a_1, \ldots, a_{\log(n)}) \rightarrow v_{a'} = y_{a'} \oplus x\).
In total, we have \(O(n)\) clauses for each assignment \(a \in \{0, 1\}^{\log(n)}\), hence, the resulting CNF has \(O(n^2)\) clauses.

Based on the values of  $u=\langle u_1,...,u_{\log n}\rangle$, $x$ and all $v_i$ for $i\in[n]$, there is a single way to pick the value of all $y_i$s for $i\in[n]$ to make $\Gamma$ true and the number of possible Skolem functions for $x$ (i.e. $S_x(u_1, \ldots, u_{\log(n)})$) is \(2^n\).
Since $\Gamma$ has $n+\log n$ $\forall$-variables and $2^n$ winning strategies, it requires exponential-size proofs in both the naive system and  \emice  (Corollary~\ref{cor_exp}). 

\vspace{0.2cm}
A linear length and polynomial size \qmice proof for $\Gamma$ is as follows (for ease, we represent $\langle u_1...u_{\log n}\rangle$ as $\vec{u}$):
For each $\beta\in[n]$, derive the following (below $i\in[n]\setminus\beta$):
{\small
\begin{prooftree}
\def\defaultHypSeparation{\hskip .05in}
\AxiomC{~} 
\UnaryInfC{$(\Gamma|_{\vec{u}=\beta},~ \big\{x=0, y_\beta=\overline{v}_\beta,y_i=v_i\big\},~ 1)$}
   \UnaryInfC{$(\Gamma|_{\vec{u}=\beta},~ \{x=0\},~ 1)$}
\AxiomC{~} \LeftLabel{(Axiom)~}
\UnaryInfC{$(\Gamma|_{\vec{u}=\beta},~ \big\{x=1, y_\beta=v_\beta,y_i=\overline{v}_i \big\},~ 1)$} \LeftLabel{(Comp-b)}
   \UnaryInfC{$(\Gamma|_{\vec{u}=\beta},~ \{x=1\},~ 1)$} \LeftLabel{(Comp-c)}
   \BinaryInfC{$L_\beta':=(\Gamma|_{\vec{u}=\beta},~ \emptyset,~ 2)$}
\end{prooftree}
}
We used composition-b and composition-c rules in the above derivation. These rules require an FQBF proof of the corresponding absence of models statements. We next describe the absence of model statement for the composition-b rules.
For ease, we represent $\langle v_1...v_{n}\rangle$ as $\vec{v}$ below.
For each $\beta\in[n]$ and $q\in\{0,1\}$, we give the absence of models statement $\Psi_{\beta,x=q}$ for the Composition-b rule below (for ease, we only show useful clauses in the matrix):\\
$\Psi_{\beta,x=q}:=\exists \vec{\alpha}_{1} ~\exists \vec{u} ~\exists x ~\forall \vec{v} ~\exists y_1,...,y_n~.~  \big(\vec{u}=\beta\rightarrow(v_\beta = y_\beta\oplus \overline{x})\land {\bigwedge}_{i\in[n]\setminus \beta}(v_i = y_i\oplus x) \big)\\\hspace*{1.5cm} \land (x = q) \land(\vec{u} = \beta) \land \Big((\vec{v}=\vec{\alpha}_1)\rightarrow (y_\beta\not= \alpha_{1,v_\beta}\oplus \overline{x})\lor {\bigvee}_{i\in[n]\setminus \beta}(y_i\not= \alpha_{1,v_i}\oplus x)  \Big)$

FQBF proof of $\Psi_{\beta,x=q}$ in Frege+$\forall$red proof system (after propagating $x=q$ and $\vec{u}=\beta$) is given below. For simplicity we denote $\alpha_{1,v_i}$ as $\alpha_i$ and the soundness of inference-lines $A,B,C$ are explained after the derivation:
{\small
\begin{prooftree}
\def\defaultHypSeparation{\hskip .02in}
\AxiomC{$(v_\beta = y_\beta\oplus \overline{q})\land \underset{i\in[n]\setminus \beta}{\bigwedge}(v_i = y_i\oplus q)$}
\AxiomC{$\Big((\vec{v}=\vec{\alpha}_1)\rightarrow (y_\beta\not= \alpha_\beta\oplus \overline{q})\lor \underset{i\in[n]\setminus \beta}{\bigvee}(y_i\not= \alpha_i\oplus q)  \Big)$} \RightLabel{(A)}
\BinaryInfC{$\Big((\vec{v}=\vec{\alpha}_1)\rightarrow (v_\beta\oplus\overline{q} \not= \alpha_\beta\oplus \overline{q})\lor \underset{i\in[n]\setminus \beta}{\bigvee}(v_i\oplus q \not= \alpha_i\oplus q)  \Big)$} \RightLabel{(B)}
\UnaryInfC{$\Big((\vec{\alpha}_1=\vec{\alpha}_1)\rightarrow (\alpha_\beta\oplus\overline{q} \not= \alpha_\beta\oplus \overline{q})\lor \underset{i\in[n]\setminus \beta}{\bigvee}(\alpha_i\oplus q \not= \alpha_i\oplus q)  \Big)$} \RightLabel{(C)}
\UnaryInfC{$\bot$}
\end{prooftree}
}

In this FQBF proof, we are only left with the soundness of inference lines A,B and C.
For $(A)$, the hypothesis is the input QBF and the conclusion is derived by propositional Frege inferences which substitute $y_i,y_\beta$ variables with the functions $v_i\oplus q, v_\beta\oplus \overline{q}$ respectively.
To derive $(B)$, since $\vec{v}$ are the rightmost variables in the hypothesis clause, the conclusion is derived by a $\forall$red rule which substitutes $\vec{v}=\vec{\alpha_1}$.
Finally, the hypothesis in $(C)$ is a contradiction as it's equivalent to true$\rightarrow$ false, allowing us to derive \(\bot\).

\vspace{0.2cm}
\noindent For each $\beta\in[n]$, the absence of models statement for the Composition-c rule is :

$\Psi_\beta:=\exists \vec{\alpha}_1, \vec{\alpha}_2 ~\exists \vec{u} ~\exists x ~\forall \vec{v} ~\exists y_1,...,y_n~.~ \chi_n \land(\vec{u}=\beta) \land \big( (\vec{v}=\vec{\alpha}_1)\rightarrow (x\not=0) \big) \land \big( (\vec{v}=\vec{\alpha}_2)\rightarrow (x\not=1) \big)$.
Since $\vec{v}\geq_{\Psi_\beta} x$, the Frege+$\forall$red proof of $\Psi_\beta$ is easy: use $\forall$red rule to substitute $\vec{v}=\vec{\alpha}_1$ and $\vec{v}=\vec{\alpha}_2$ in the two axiom clauses to derive $(x\not=0)$ and $(x\not=1)$, a contradiction.

\noindent So far in \qmice proof, we have derived lines $L_\beta'$ for every $\beta\in[n]$. Now we can apply the join rule $n-1$ times and remove the restriction of universal variables $u_{\log n}$ to $u_1$ (in this order). This results in the last claim being $(\Gamma,~\emptyset,~ 2^n)$.

The two upper bounds describe above, put together, allow us to prove the following separation between the proof systems we have considered.

\begin{theorem}\label{seperation_theorem}
    The \#QBF proof system \qmice is exponentially separated from the naive \#QBF system and the \emice proof system.
\end{theorem}

\section{Discussion and Future Work}
The paper proposes a \#QBF proof system \qmice based on the \#SAT proof system \mice~\cite{fichte2022proofs} and proves that it is exponentially stronger than the naive (Section~\ref{naive_sec}) and \emice proof systems (Section~\ref{sec:forexp_mice}). For this, we introduced two new families of true QBFs: the quantified XOR-PAIRS and the indexed affine QBFs. We give easy \qmice proofs for both the formulas, and show that they are hard for the naive and \emice proof systems. 

Recently, a \#QBF solver d4-QBF~\cite{CapelliLPS24} has been introduced. One open problem is to compare the strength of \qmice and the d4-QBF solver. 
It is easy to observe that \qmice is capable of certifying all the rules of d4-QBF except the decomposition rule. However, \qmice is modular, in the sense that, based on the requirements, one can always add more sound \#QBF inference rules like the decomposition rule from~\cite{CapelliLPS24}. One way to incorporate the decomposable rule is to add it at the beginning of a \qmice proof whenever possible. To be precise,
suppose $Q.\phi$ is the input QBF. Run the `connected-component’ function of d4-QBF, and let it returns, $Q_1.\phi_1, Q_2.\phi_2, Q_3.\phi_3$ with disjoint variables.  Since \qmice is complete, running it on $Q_1.\phi_1$ will eventually derive the claim $(Q_1.\phi_1,\null, c_1)$ and similarly, the claims with $c_2, c_3$. Finally, the model-count of $Q.\phi$ would be $c_1\cdot c_2\cdot c_3$ owing to the correctness of the decomposition step.

The challenge of doubly exponential solutions may be mitigated in \qmice by representing integers by arithmetic circuits. 
An immediate open problem is to establish a genuine lower bound~\cite{DBLP:conf/stacs/BeyersdorffB18} for \qmice. We conjecture that the propositional XOR-PAIRS (which are hard for MICE~\cite{BeyersdorffHS24}) would be hard for \qmice but it is not a genuine lower bound.
Another open problem is to extend other \#SAT proof systems \cpog~\cite{bryant2023certified}, \clip~\cite{ChedeCS24}, \kcps~\cite{Capelli19} for \#QBF.

\bibliography{bib2doi}

\end{document}